\documentclass[11pt,a4paper]{amsart}
\usepackage{amssymb} 
\usepackage{amsmath}
\usepackage{graphicx}
\usepackage{amsthm}
\newtheorem{theorem}{Theorem}[section]
\newtheorem{lemma}[theorem]{Lemma}
\newtheorem{obs}[theorem]{Observation}
\newtheorem{prop}[theorem]{Proposition}

\long\def\onefigure#1#2{
\begin{figure*}[tbp]
\begin{center}
#1
\end{center}
\caption{#2}
\end{figure*}
}

\newcommand{\labepsfig}[2]  
{\onefigure{\mbox{\includegraphics{#1.eps}}}{\label{f:#1}#2}}

\newcommand{\eps}{\varepsilon}

\newcommand{\tuple}{\mathbf}

\newcommand{\DD}{\tuple D}

\newcommand{\PP}{\tuple P}
\newcommand{\RR}{\tuple R}
\renewcommand{\SS}{\tuple S}

\newcommand{\VV}{\tuple V}

\DeclareMathOperator{\conv}{conv}
\newcommand{\Nset}{\mathbb N}
\newcommand{\Rset}{\mathbb R}
\newcommand{\R}{\Rset}
\DeclareMathOperator{\Dom}{\mathbf{Dom}}
\DeclareMathOperator{\dom}{dom}
\newcommand{\closure}{\overline}
\newcommand{\bd}{\partial}
\newcommand{\dist}{\mathrm{dist}}
\newcommand{\cBall}{\mathrm B}
\newcommand{\support}{\top}

\begin{document}


\begin{center}
\leavevmode\bigskip\bigskip

\textbf{ZONE DIAGRAMS IN EUCLIDEAN SPACES\\AND IN OTHER NORMED SPACES}

\bigskip
\bigskip
\bigskip

\textsc{Akitoshi Kawamura}%
\footnote{%
Part of this work was done while A.K. was visiting 
ETH Z\"urich, 
whose support and hospitality are gratefully acknowledged. 
His research is also supported by
the Nakajima Foundation and 
the Natural Sciences and Engineering Research Council of Canada.
}%
\\[2pt]
\begin{footnotesize}
Department of Computer Science, 
University of Toronto \\
10 King's College Road, 
Toronto, Ontario, M5S\,3G4 Canada\\
\texttt{kawamura@cs.toronto.edu}\par
\end{footnotesize}

\bigskip

\textsc{Ji\v{r}\'{\i} Matou\v{s}ek}%
\\[2pt]
\begin{footnotesize}
Department of Applied Mathematics and\\
Institute of Theoretical Computer Science (ITI), 
Charles University\\
Malostransk\'{e} n\'{a}m.~25, 
118\,00~~Praha~1, Czech Republic, and
\\
Institute of Theoretical Computer Science, ETH Z\"urich\\
8092 Z\"urich, Switzerland\\
\texttt{matousek@kam.mff.cuni.cz}\par
\end{footnotesize}

\bigskip

\textsc{Takeshi Tokuyama}%
\footnote{%
The part of this research by T.T. was partially
supported by the 
JSPS Grant-in-Aid for Scientific Research (B) 18300001. 
}%
\\
\begin{footnotesize}
Graduate School of Information Sciences, Tohoku University\\
Aramaki Aza Aoba, Aoba-ku, Sendai, 
980-8579 Japan\\
\texttt{tokuyama@dais.is.tohoku.ac.jp}\par
\end{footnotesize}
\end{center}

\bigskip
\bigskip

\begin{small}\leftskip21pt\rightskip\leftskip
\noindent
\textsc{Abstract. }
Zone diagram is a variation on the classical concept of
a Voronoi diagram. Given $n$ sites in a metric space that compete
for territory, the zone diagram is an equilibrium state 
in the competition. Formally it is defined as a fixed point of a certain
``dominance'' map. 

Asano, Matou\v{s}ek, and Tokuyama proved
the existence and uniqueness of a zone diagram for point
sites in Euclidean plane, and Reem and Reich showed existence
for two arbitrary sites in an arbitrary metric space.
We establish existence and uniqueness for $n$ disjoint compact
sites in a Euclidean space of arbitrary (finite) dimension,
and more generally, in a finite-dimensional normed space
with a smooth and rotund norm. The proof is considerably simpler
than that of Asano et al. We also provide an example of non-uniqueness
for a norm that is rotund but not smooth. Finally,
we prove existence and uniqueness for two point sites in the plane
with a smooth (but not necessarily rotund) norm.
\par
\end{small}

\thispagestyle{empty}


\clearpage
\setcounter{page}{1}
\markboth{ZONE DIAGRAMS IN NORMED SPACES}{KAWAMURA, MATOU\v SEK, TOKUYAMA}

\section{Introduction}

Zone diagram is a metric notion somewhat similar to the classical
concept of a Voronoi diagram. Let $(X,\dist)$ be a metric space
and let $\PP=(P_1,\ldots,P_n)$ be an $n$-tuple of
nonempty subsets of $X$ called the \emph{sites}. To avoid unpleasant
trivialities, we will always assume in this paper that the sites 
are closed and pairwise disjoint.

A \emph{zone diagram}
of the $n$-tuple $\PP$ is an $n$-tuple $\RR=(R_1,\ldots,R_n)$
of subsets of $X$, called the \emph{regions} of the zone diagram,
with the following defining property: Each $R_i$ 
consists of all points $x\in X$ that are closer (non-strictly)
to $P_i$ than to the union $\bigcup_{j\ne i}R_j$ of all
the other regions. 

Fig.~\ref{f:example-zod} shows a zone diagram
in Euclidean plane whose sites are points and segments.
While in the Voronoi diagram the regions partition the whole space,
in a zone diagram the union of the regions typically has a nonempty
complement, called the \emph{neutral zone}. 

\labepsfig{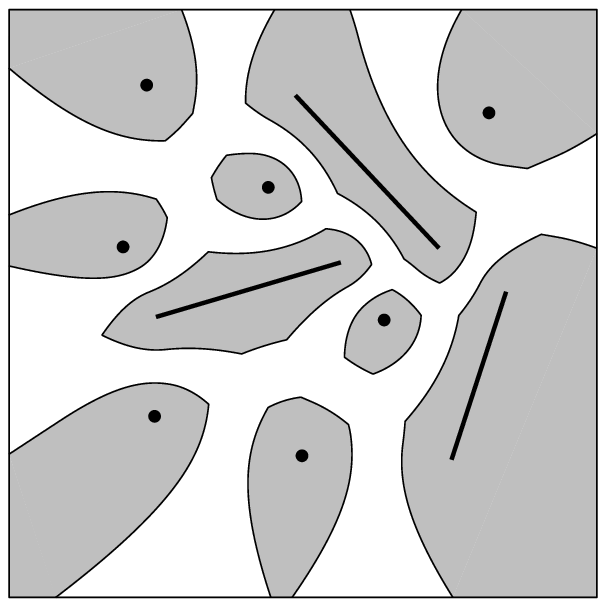}{A zone diagram of points and segments.}

The definition of the zone diagram is  implicit, since each region is 
determined  in terms of the remaining ones. So neither existence 
nor uniqueness of the zone diagram is obvious, and so far only
partial results in this direction have been known.

Asano et al.~\cite{asano07:_zone_diagr} introduced the notion
of a zone diagram, for the case of $n$ point sites in Euclidean plane,
and in this setting they proved existence and uniqueness.
The proof involves a case analysis specific to $\Rset^2$.

Reem and Reich \cite{reem07:_zone_and_doubl_zone_diagr}
established, by a simple and elegant argument,
the existence of a zone diagram for \emph{two} sites
in an arbitrary metric space (and even in a still more general setting,
which they call \emph{$m$-spaces}).

On the negative side, they gave an example of a three-point metric space
in which the zone diagram of two point sites is not unique; 
thus, \emph{uniqueness} needs  additional assumptions.
On the other hand, for all we know, 
it is possible that a zone diagram always exists, 
for arbitrary sites in an arbitrary metric space.

\subsection*{Arbitrary sites in Euclidean spaces. } 
In this paper, 
we establish the existence and uniqueness of zone diagrams
in Euclidean spaces. 
This generalizes the main result of \cite{asano07:_zone_diagr}
with a considerably simpler argument. For the case of two
point sites in the plane, we also obtain a new and simpler proof of
the existence and uniqueness of the \emph{distance trisector
curve} considered by Asano et al.~\cite{asano07:_distan_trisec_curve}.

\begin{theorem} \label{t:euclcase}
Let the considered metric space $(X,\dist)$ 
be $\R^d$ with the Euclidean distance.
For every $n$-tuple $\PP=(P_1,\ldots,P_n)$ of nonempty closed sites in $\R^d$
such that $\dist(P_i,P_j)>0$ for every $i\ne j$, there exists
exactly one zone diagram~$\RR$.
\end{theorem}

The full proof is contained in Sections~\ref{s:prelim} (general
preliminaries) and~\ref{s:euclcase}. The same proof yields
existence and uniqueness also for infinitely many sites in $\R^d$,
provided that every two of them have distance at least $1$
(or some fixed $\eps>0$). Moreover, with some extra effort
it may be possible to extend  the proof to compact sites in a Hilbert
space, for example, but in this paper we restrict ourselves to
the finite-dimensional setting.

\subsection*{Normed spaces. } We also investigate zone diagrams in
a more general class of metric spaces, namely, finite-dimensional
normed spaces.\footnote{A finite-dimensional (real) normed space
can be thought of as the real vector space $\R^d$ with some \emph{norm},
which is a mapping  that assigns a nonnegative real number $\|x\|$
 to each $x\in\R^d$ so that
$\|x\|=0$ implies $x=0$, $\|\alpha x\|=|\alpha|\cdot\|x\|$
 for all $\alpha\in\R$,
and the triangle inequality holds: $\|x+y\|\leq \|x\|+\|y\|$.
Each norm $\| \mathord\cdot \|$ defines a metric by
$\dist(x,y):=\|x-y\|$.

For studying a norm $\| \mathord\cdot \|$, it is usually good to look
at its \emph{unit ball} $\{\, x\in\R^d: \|x\|\le 1 \,\}$.
The unit ball of any norm is a closed convex body $K$ that is
symmetric about $0$ and contains $0$ in the interior.
Conversely, any $K\subset\R^d$ with the listed properties is the unit ball
of a (uniquely determined) norm.
}
Normed spaces are among the most important classes of metric spaces.
Moreover, as we will see, studying arbitrary norms also sheds some
light on the Euclidean case. Earlier Asano and Kirkpatrick
\cite{AsanoKirkpatrick} investigated distance trisector
curves (which are essentially equivalent
to two-site zone diagrams)
of two point sites under polygonal norms in the plane,
obtaining results for the Euclidean case through approximation
arguments.

For us, a crucial observation is that the uniqueness of zone diagrams
does \emph{not} hold for normed spaces. Let us consider
$\R^2$ with the $\ell_1$ norm $\| \mathord\cdot \|_1$, given by
$\|x\|_1=|x_1|+|x_2|$. It is easy to check that 
the two point sites $(0,0)$ and $(0,3)$ have at least two 
different zone diagrams, as drawn in Fig.~\ref{f:l1nonunique}.
This example was essentially contained already in Asano and Kirkpatrick
\cite{AsanoKirkpatrick}, although in a different context.

\labepsfig{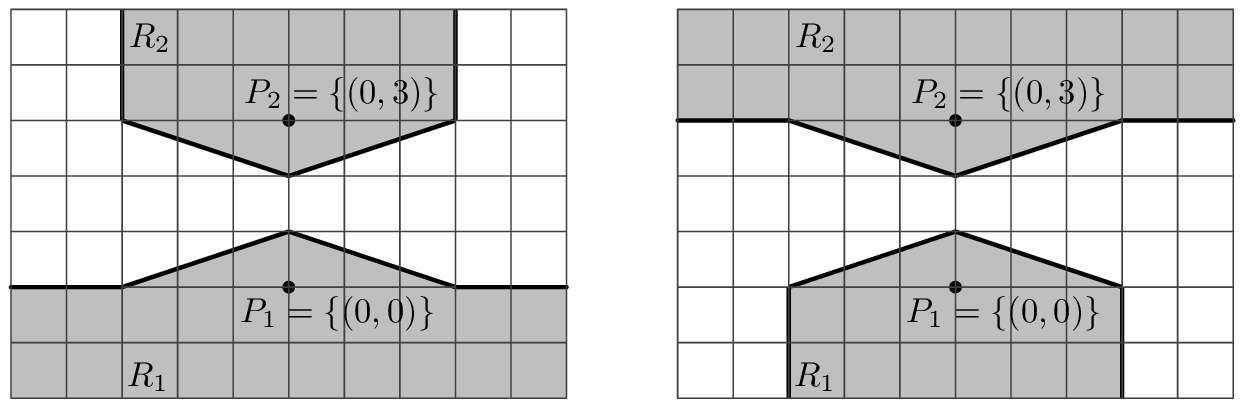}{Two different zone diagrams 
under the $\ell_1$ metric (drawn in the grid with unit spacing). }

The $\ell_1$ norm differs from the Euclidean norm in two basic
respects: the unit ball has sharp corners and
straight edges; in other words, the $\ell_1$ norm is neither 
smooth nor rotund. We recall that
a norm $\| \mathord\cdot \|$ on $\R^d$ is called \emph{smooth} if the function
$x\mapsto \|x\|$ is differentiable (geometrically,
the unit ball of a smooth norm has no ``sharp corners'';
see Fig.~\ref{f:norms}).\footnote{There
are several notions of differentiability of functions
on Banach spaces, such as the existence of directional
derivatives, G\^ateaux differentiability,
Fr\'echet differentiability, or
uniform Fr\'echet differentiability. However,
in finite-dimensional Banach spaces they
are all equivalent.} 
A  norm $\| \mathord\cdot \|$ on $\R ^d$ is called \emph{rotund}
(or \emph{strictly convex}) if
for all $x,y\in\R^d$ with $\|x\|=\|y\|=1$
and $x\ne y$ we have
$\|\frac{x+y}2\|<1$. Geometrically, the unit sphere
of $\| \mathord\cdot \|$ contains no segment.
By compactness, a rotund norm on a finite-dimensional space is also \emph{uniformly convex}, which means
that for every $\eps>0$ there is $\mu=\mu(\eps)>0$ such that if $x,y$ are 
unit vectors with $\|x-y\|\ge\eps$, then
\begin{equation*}\label{e:unifconv}
\left\|\frac{x+y}2 \right\|\le1-\mu 
\end{equation*}
(we refer to \cite{BenyaminiLindenstrauss}
for this and other facts on norms mentioned without proofs).

\labepsfig{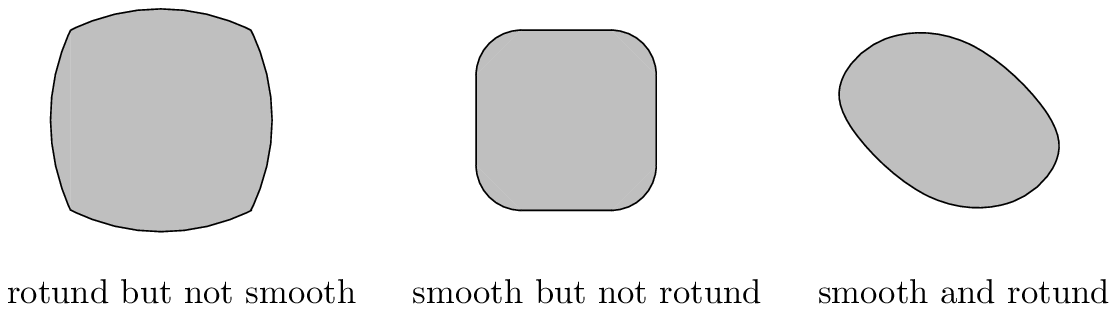}{Rotundity and smoothness of norms.}

The Euclidean norm $\| \mathord\cdot \|_2$, and more generally, the
$\ell_p$ norms with $1<p<\infty$, are both rotund and smooth.
We have the following generalization of Theorem~\ref{t:euclcase}:

\begin{theorem} \label{t:norms}
Let the considered metric space $(X,\dist)$
be $\R^d$ with  a norm $\| \mathord\cdot \|$ that is both smooth and rotund.
For every $n$-tuple $\PP=(P_1,\ldots,P_n)$ of nonempty closed sites in $\R^d$
such that $\dist(P_i,P_j)>0$ for every $i\ne j$, there exists
exactly one zone diagram~$\RR$.
\end{theorem}

The proof for the Euclidean case, i.e., of Theorem~\ref{t:euclcase},
is set up so that it generalizes to smooth and rotund norms more
or less immediately; there is only one lemma where we need to
work harder---see Section~\ref{s:nrmproof}.

Our current proof method apparently depends both on smoothness
and on rotundity. In Section~\ref{s:nonuniq} we show that
smoothness is indeed essential, by exhibiting a non-smooth
but rotund norm in $\R^d$ with non-unique zone diagrams.
On the other hand, we suspect that the assumption
of rotundity in Theorem~\ref{t:norms} can be dropped.
Currently we have a proof (see Appendix~\ref{s:smoothonly}) 
only in a rather special case:

\begin{theorem}
\label{theorem: smooth two singleton plane}
For two point sites $P _0 = \{p _0\}$ and $P _1 = \{p _1\}$
in the plane $\Rset ^2$ with a smooth norm, 
there exists exactly one zone diagram. 
\end{theorem}

\section{Preliminaries}\label{s:prelim}

Here we introduce notation and present some results from the literature,
some of them in a more general context than in the original publications.

Let $(X,\dist)$ be a general metric space.
The closure of a set $A\subseteq X$ is denoted by
$\closure{A}$, while $\bd A$ stands for its boundary.
The (closed) ball of radius $r$ centered at $x$
is denoted by $\cBall(x,r)$. 

For sets $A, B \subseteq X$, not both empty,
we define the \emph{dominance region} of $A$ over $B$
as the set
$$
\dom(A,B):=\{\, x\in X: \dist(x,A)\le \dist(x,B) \,\},
$$
where  
$$
\dist(C,D) := \inf _{x \in C, \ y\in D} \dist(x, y) \in [0, +\infty]
$$ 
denotes the distance of sets $C$ and $D$.

Let us fix an $n$-tuple $\PP=(P_1,\ldots,P_n)$ of 
sites, i.e., nonempty subsets of $X$ (which, as above,
we assume to be disjoint and closed).
For an $n$-tuple $\RR=(R_1,\ldots,R_n)$ of arbitrary
subsets of $X$, we define another $n$-tuple of regions
$\RR'=(R'_1,\ldots,R'_n)$ denoted by $\Dom\RR$ and given by
$$
R'_i:=\dom \Bigl( P _i, \bigcup _{j \neq i} R _j \Bigr), 
 \qquad 
  i = 1, \dots, n
$$
(the sites are considered fixed and they are a part of the definition
of the operator $\Dom$).

The definition of a zone diagram can now be expressed as follows:
An $n$-tuple $\RR$ is called a \emph{zone diagram} for the $n$-tuple
$\PP$ of sites if $\RR=\Dom\RR$ (componentwise equality,
i.e., $R_i=\dom\bigl(P _i, \bigcup _{j \neq i} R _j\bigr)$ for all $i$).

For two $n$-tuples $\RR$ and $\SS$ of sets, we write $\RR\preceq \SS$
if $R_i\subseteq S_i$ for every $i$. 
It is easily seen (see, e.g., \cite{asano07:_zone_diagr}) that the operator $\Dom$
is antimonotone, i.e., $\RR\preceq\SS$ implies $\Dom \RR\succeq \Dom\SS$.
Our starting point in the proof
of Theorems~\ref{t:euclcase} and~\ref{t:norms}
is the following general result
(see Appendix~\ref{appendix section: double zone diagrams} for a proof):

\begin{theorem}[{\cite[Lemma~5.1]{asano07:_zone_diagr}}, \relax{\cite[Theorem~5.5]{reem07:_zone_and_doubl_zone_diagr}}]
 \label{t:doublezone}
For every $n$-tuple $\PP$ of sites (in any metric space)
there exist $n$-tuples $\RR$ and $\SS$ 
such that $\RR=\Dom\SS$ and $\SS=\Dom \RR$. Moreover, for every
$n$-tuples $\RR',\SS'$ with $\RR'=\Dom\SS'$ and $\SS'=\Dom \RR'$
we have $\RR\preceq \RR',\SS'\preceq \SS$ 
(and in particular, $\RR\preceq\SS$).
\end{theorem}

We finish this section with a simple geometric lemma.
It was used, in a less general setting, in \cite{asano07:_zone_diagr}
(proof of Lemma~4.3).

\begin{obs}\label{o:eps4ball}
Let $\PP$ be an $n$-tuple of sites (in an arbitrary metric space), and
suppose that $\varepsilon := \min _{i \neq j} \dist (P _i, P _j) > 0
$ and that $\RR$ and $\SS$ satisfy $\RR=\Dom\SS$ and $\SS=\Dom\RR$.
Then $\dist (P _i, \bigcup _{j \neq i} S _j) \geq \frac \eps 2$,
and consequently, the $\frac\eps4$-neighborhood
of each $P_i$ is contained in $R_i$.
\end{obs}

\begin{proof}
We recall the simple proof from
\cite{asano07:_zone_diagr}.
We first note that $\VV=(V _1, \dots, V _n) := \Dom \PP$ is
the classical Voronoi diagram of $\PP$,
and the open $\frac\eps2$-neighborhood of $P_i$
does not intersect $\bigcup_{j\ne i} V_j$.
Since $\PP\preceq \RR$, we have $\Dom \PP \succeq \Dom \RR = \SS$, 
and hence the open $\frac\eps2$-neighborhood of $P_i$
is disjoint from $\bigcup_{j\ne i} S _j$ as well, as claimed.
\end{proof}

\section{The Euclidean case}\label{s:euclcase}

Here we prove Theorem~\ref{t:euclcase};
throughout this section, $\dist$ denotes the Euclidean
distance. In addition
to Theorem~\ref{t:doublezone} and Observation~\ref{o:eps4ball},
we also need the next lemma.

\begin{lemma}[Cone lemma, Euclidean case] \label{l:eucl-cone}
Let $\PP$ be an $n$-tuple of (nonempty closed) sites in $\R^d$ with
the Euclidean metric with
$\varepsilon:=\min_{i\ne j} \dist (P _i, P _j) > 0$,
and let $\RR$ and $\SS$ satisfy $\RR=\Dom\SS$ and $\SS=\Dom\RR$.
Let $a$ be a point of some $R_i$, and let $p\in P_i$
be a point of the corresponding site closest to $a$ (such a nearest
point exists by compactness).
Then the set
\begin{equation*}
 K := \conv \bigl( \{a\} \cup \cBall (p,{\textstyle\frac\eps4}) \bigr)
\end{equation*}
is contained in $R_i$; see Fig.~\ref{f:coneK}.
\labepsfig{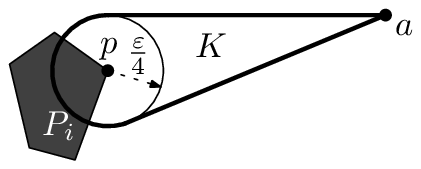}{The cone $K$.}
\end{lemma}

The following proof is rather specific for the Euclidean metric 
(the lemma fails for the $\ell_1$ metric, for example).

\begin{proof}
Both $a$ and $\cBall (p,\frac\eps4)$ are contained
in $\dom(p,\bigcup_{j\ne i}S_j)$ (the latter by Observation~\ref{o:eps4ball}).
For the Euclidean metric, the dominance region of a point over
any set is convex, since it is the intersection of
halfspaces. Hence $K\subseteq \dom(p,\bigcup_{j\ne i}S_j)\subseteq R_i$.
\end{proof}

Now we describe the general strategy of the proof of Theorem~\ref{t:euclcase}.
With $\RR$ and $\SS$ as in Theorem~\ref{t:doublezone}, 
it suffices to prove $\RR=\SS$. 
For contradiction, we assume that it is not the case,
i.e., that $R:=\bigcup_{i=1}^n R_i$ is properly contained in 
$S:=\bigcup_{i=1}^n S_i$; see
the schematic illustration in Fig.~\ref{f:RandS}.

\labepsfig{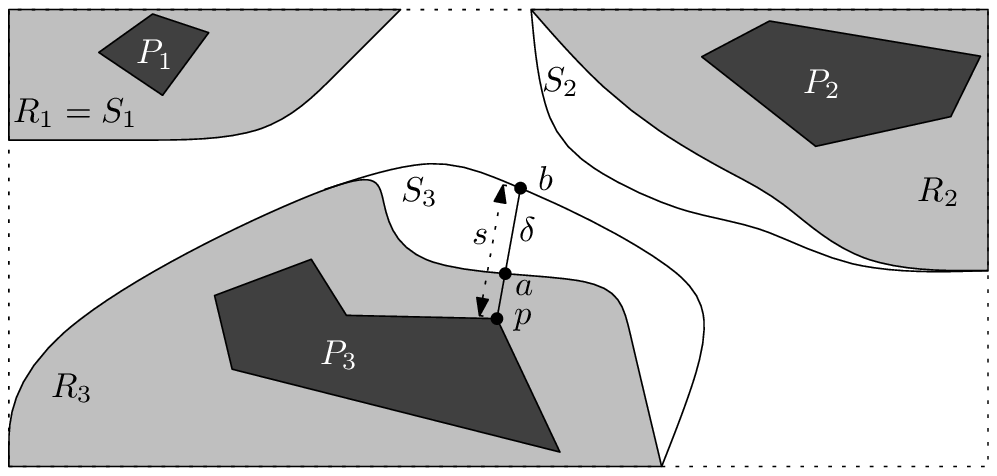}{The setting of the proof of Theorem~\ref{t:euclcase}
(a schematic picture).}

For a point $b \in S \setminus P$,
let $s (b) := \dist (b, P)$ be the distance from the nearest site,
and let $p = p (b) \in P _i$ be a point where this distance is attained.
Let $a = a (b)$ be the closest point to $b$ 
that lies in the intersection of $R _i$ with the segment~$b p$.
It is easily seen, using the triangle inequality,
that $p$ is also a nearest point of $P$ to $a$.
Thus, the set $K$ in Lemma~\ref{l:eucl-cone} is contained in $R _i$,
and in particular, $a$ is the only intersection
of the segment $bp$ with $\bd R_i$.
We set $\delta (b) := \dist(b,a)$.
The parameters $s(b)$ and $\delta(b)$ will measure, in some sense,
how much $S$ differs from $R$ ``at $b$''.

Assuming $\RR \neq \SS$, we choose a point $b _0 \in S \setminus R$.
Then, using $b_0$, we find $b_1\in S\setminus R$ where
$S$ differs from $R$ ``more than'' at $b_1$. Iterating the same procedure
we obtain an infinite sequence $b_0,b_1,b_2,b_3,\ldots$ of points,
and the difference will ``grow'' beyond bounds, while, on the other hand,
it has to stay bounded---and this way we reach a contradiction.

More concretely, for every integer $t\ge 1$ we will construct
$b_t$ from $b_{t-1}$ so that, with $s:=s(b_{t-1})$,
$s':=s(b_t)$, $\delta:=\delta(b_{t-1})$, and $\delta':=\delta(b_t)$,
we have
\begin{enumerate}
\def\theenumi{\textup{(\Alph{enumi})}} 
\def\labelenumi{\theenumi}
\item \label{enumi: length decreases}
$s'\leq s - \alpha$, {or}
\item \label{enumi: gap increases}
$s' \leq s - \delta$ and 
$\delta' \geq \delta$,
\end{enumerate}
where $\alpha > 0$ is a constant that
depends on $s _0 := s (b _0)$ and $\eps$,  but not on $t$.

Thus, as $t$ increases, $s (b _t)$ keeps decreasing.
Since $s (b _t)$ is bounded from below by $\frac\eps4$
by Observation~\ref{o:eps4ball},
case~\ref{enumi: length decreases} can happen only finitely many times.
Therefore, from some $t$ on,
we have case~\ref{enumi: gap increases} only.
But this also causes $s (b _t)$ to decrease towards $0$---a contradiction.

It remains to describe the construction of $b_{t}$ from $b_{t-1}$,
and this is done in the next lemma.

\begin{lemma}\label{l:step}
For every $s _0$ and $\eps>0$ there exists $\alpha > 0$
such that if $b \in S \setminus R$ satisfies $s := s (b) \leq s _0$,
then there exists another point $b' \in S \setminus R$ such that
$s' := s (b')$, $\delta := \delta (b)$ and $\delta' := \delta (b')$ satisfy
\ref{enumi: length decreases} or \ref{enumi: gap increases}.
\end{lemma}

\begin{proof}
Let $b \in S _i$, let $a := a(b)$, $p:=p(b)$,
and write $r = \dist(a,p)$; see Fig.~\ref{f:newb1}.
Since $a \in \bd R_i$ and $\RR = \Dom \SS$,
there exist $j \neq i$ and $b' \in S _j$ with $\dist(a,b')=r$.
If there are several possible $b'$, we choose one of them arbitrarily.
\labepsfig{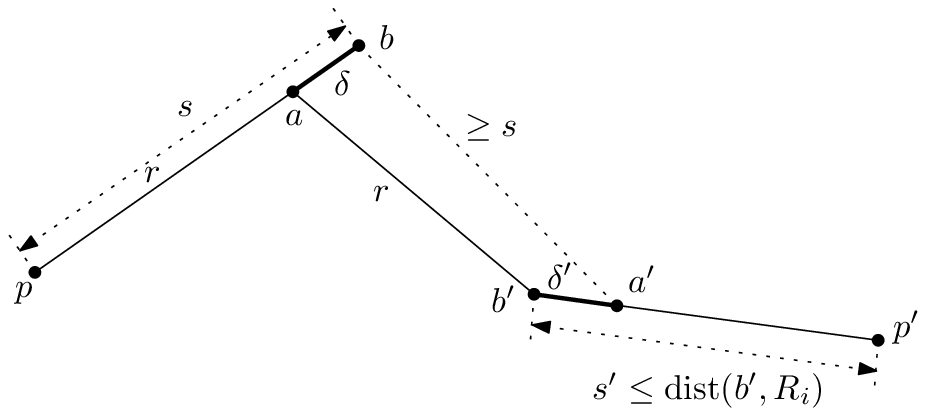}{The construction of $b'$.}

First we check that $b'\not\in R$, or in other words, that
$\delta'>0$. During this step we also derive a lower bound for $\delta'$ that
will be useful later. Since $b\in S$, $a'\in R$, and $\SS=\Dom\RR$,
we have $\dist(a',b)\ge s$. Then we bound, using the triangle inequality,
\begin{equation}\label{e:delta'}
\delta'\ge \dist(a',b)-\dist(b,b')\ge s-\dist(b,b').
\end{equation}
Supposing for contradiction that $\delta'=0$, we get
$\dist(b,b')=s$. But the triangle inequality
gives $\dist(b,b') \le \dist(b,a)+\dist(a,b')=r+\delta=s$,
and hence the triangle inequality here holds with equality.
For the Euclidean metric, this can happen only if $a$ lies
on the segment $bb'$,
and then $b'$ has to coincide with $p$, which is impossible.
So $\delta'>0$ indeed.

Next, since $\SS=\Dom\RR$ and $b'\in S$, 
we have $s'\le \dist(b',R_i)$. An obvious upper bound
for $\dist(b',R_i)$ is $\dist(b',a)=r=s-\delta$, and thus
the first inequality in \ref{enumi: gap increases},
namely, $s'\le s-\delta$, always holds.  

Moreover, if $\delta\ge\alpha$,
then $s'\le s-\delta\le s-\alpha$,
and we have \ref{enumi: length decreases}.
For the rest of the proof we thus assume that $\delta<\alpha$
(where $\alpha$ hasn't been fixed yet---so far we're
free to choose it as a positive function of $\eps$ and $s_0$
in any way we like).

Let us consider the ball $\cBall (b',r)$; see Fig.~\ref{f:cutball}.
If it contains $b$, as in the left picture, we have
$\dist(b',b)\le r$, and thus by (\ref{e:delta'}) we have
$\delta'\ge s-r=\delta$. Then \ref{enumi: gap increases} holds.
Thus, the last case to deal with is
$b\not\in \cBall (b',r)$.

\labepsfig{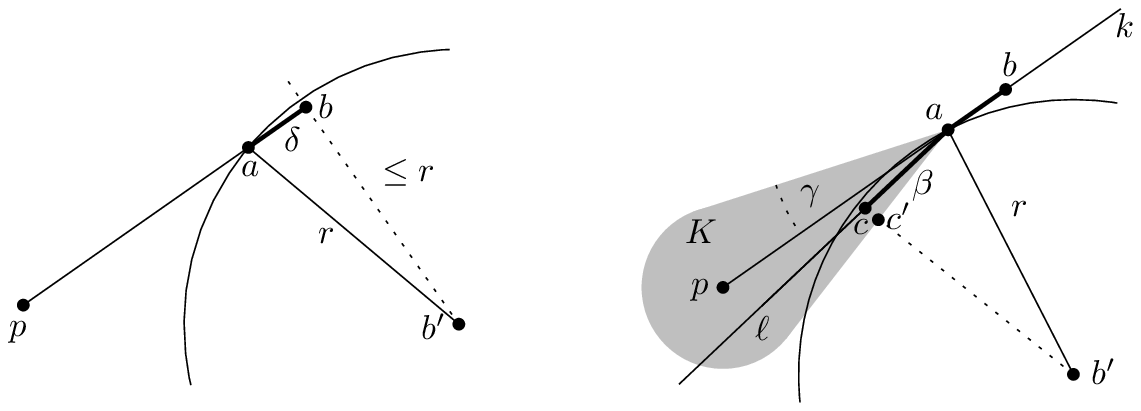}{The $r$-ball around $b'$.}

Let us consider the cone $K=\conv(\{a\}\cup \cBall(p,\frac\eps4))$
as in Lemma~\ref{l:eucl-cone}. Its opening angle $\gamma$ is
bounded away from $0$ in terms of $\eps$ and $s_0$.

Let $\Pi$ be a 2-dimensional plane containing $p,a,b'$;
it also contains $b$ since $p,a,b$ are collinear.
Let $k$ be the ray originating at $a$ and containing $b$,
and let $\ell$ be the ray in $\Pi$ originating at $a$
and making the angle $\pi-\frac\gamma2$ with $k$ (on the side of $b'$);
see Fig.~\ref{f:cutball} right. 

Since the angle of the rays $k$ and $\ell$ is bounded away from
the straight angle, the Euclidean ball $\cBall (b',r)$ cuts a segment
of significant length $\beta$ 
from at least one of these rays; here $\beta$ can be bounded
from below by a positive quantity depending only on  $s_0$ and $\eps$.
So far we haven't fixed $\alpha$, and so now we can make
sure that $\alpha<\beta$.
Since we assume $b\not\in \cBall (b',r)$, the segment of length
$\beta$ cut out by $\cBall (b',r)$ can't belong to the ray $k$.
So the situation is as in Fig.~\ref{f:cutball} right:
$\cBall (b',r)$ contains the initial segment $ac$ of $\ell$
of length $\beta$. Hence $\dist(b',c)\le r$.

The distance $\dist(c,\R^d\setminus K)$ is bounded away from $0$
in terms of $\beta$ and $\gamma$, and so we may fix $\alpha$
so that  $\dist(c,\R^d\setminus K)\ge \alpha$. 

Let $c'$ be the point where the segment $b'c$ meets the boundary of~$K$.
We have 
$$
\dist(b',K)\le \dist(b',c')=\dist(b',c)-\dist(c,c')\le
r-\dist(c,\R^d\setminus K)\le r-\alpha.
$$
Then, finally, using $K\subseteq R_i$, we have
$$
s'\le \dist(b',R_i)\le \dist(b',K) \le r-\alpha<s-\alpha,
$$
and so \ref{enumi: length decreases} holds.
This concludes the proof of Lemma~\ref{l:step},
as well as that of Theorem~\ref{t:euclcase}.
\end{proof}

\section{The case of smooth and rotund norms}\label{s:nrmproof}

In this section we establish Theorem~\ref{t:norms}.
We begin with the part where the proof differs from the Euclidean
case: the cone lemma. In the Euclidean case, we used the fact
that for points $p\ne q$, $\dom(p,q)$ is a halfspace,
and consequently, $\dom(p,X)$ is convex for arbitrary $X$.
For other norms $\dom(p,q)$ need not be convex, though;
see Fig.~\ref{f:l4bisect}.

\labepsfig{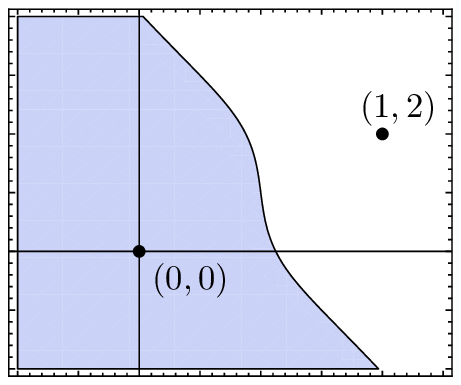}{The dominance region of the point $(0,0)$
against $(2,1)$ in the $\ell_4$ norm.}

We have at least the following convexity result.

\begin{lemma}\label{l:convcompl}
Let us consider $\R^d$ with an arbitrary norm
$\| \mathord\cdot \|$, let $H$ be a closed halfspace, and let $p \notin H$ be a point.
Then $\dom(p,H)$ is convex.

Consequently,
if the complement of a closed set $A \subseteq \Rset ^d$ is convex and 
$p \notin A$, then
$\dom (p, A)$ is convex.
\end{lemma}

\begin{proof}
Let $x \notin H$ be a point and let $x^*\in\bd H$ be
a point where $\dist (x,H)$, the distance of $x$ to $H$
measured by $\| \mathord\cdot \|$, is attained. If $y\not\in H$
is another point and $y^*\in \bd H$ is the point
such that the vectors $x-x^*$ and $y-y^*$ are parallel,
then $\|y-y^*\| = \dist (y, H)$; see Fig.~\ref{f:domhalf}.

\labepsfig{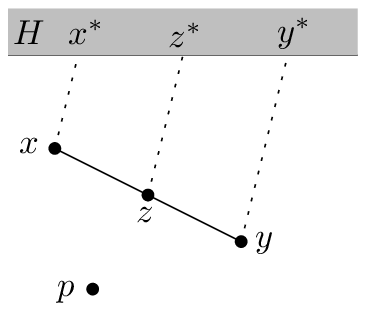}{The dominance region of a point
against a halfspace. }

Now let $x,y\in \dom(p,H)$, let $x^*,y^*$ be as above,
set $z:=(x+y)/2$, and let $z^*$ be defined analogously
to $y^*$. Then we get $\dist (z, H) = \|z-z^*\|=
(\|x-x^*\|+\|y-y^*\|)/2=(\dist(x,H)+\dist(y,H))/2$.
From this $z\in \dom(p,H)$ is immediate,
since $\|p-z\|\le (\|p-x\|+\|p-y\|)/2\le
(\dist(x,H)+\dist(y,H))/2=\dist(z,H)$. This proves the first
part of the lemma.

The second part follows easily: $A$ can be expressed
as a union of closed halfspaces $H$,
and $\dom (p, A)$ is the intersection of
the convex sets $\dom(p,H)$.
\end{proof}

Now we prove a cone lemma, similar to Lemma~\ref{l:eucl-cone}:

\begin{lemma}[Cone lemma for rotund norms] \label{lemma:cone}
Let $\| \mathord\cdot \|$ be a rotund norm on $\R ^d$.
Suppose that an $n$-tuple $\PP$ of sites satisfies
$\varepsilon:=\min_{i\ne j} \dist (P _i, P _j) > 0$,
and $\RR$ and $\SS$ satisfy $\RR=\Dom\SS$ and $\SS=\Dom\RR$.
Then for every $s_0>0$ there is $\rho>0$ (also depending
on $\eps$ and on $\| \mathord\cdot \|$) such that the following holds:
If $a\in R_i$ with $r := \dist (a, P _i) \leq s _0$
and  $p\in P_i$ is a point attaining the distance $\dist (a, P _i)$,
then the set
\begin{equation*}
 K 
:=
 \conv \bigl( \{a\} \cup \cBall (p,\rho) \bigr)
\end{equation*}
is contained in $R_i$.
\end{lemma}

\begin{proof} As in the Euclidean case,
we begin by observing that $a\in \dom(p,\bigcup_{j\ne i}S_j)$ and also
$\cBall (p,\frac\eps4) \subseteq \dom(p,\bigcup_{j\ne i}S_j)$
by Observation~\ref{o:eps4ball}.
Thus, the set
$D := \cBall (a,r) \cup \cBall (p, \frac\eps2)$ is contained in
the closure of $\Rset ^d\setminus \bigcup_{j\ne i}S_j$.
We now want to find a open convex subset  $C \subseteq D$ such that
$a$ and $B(p,\rho)$ are contained in $\dom (p, \Rset ^d \setminus C)$,
since the latter region is convex by Lemma~\ref{l:convcompl}
and thus it contains $K$ as well.

We let $C$ be the interior of $\conv (B (a, r) \cup B (p, 2 \rho))$ 
with $\rho$ sufficiently small 
(the restrictions on it will be apparent from the proof below);
see Fig.~\ref{f:twoballs}.
It is clear that $\{a\}\cup  B(p,\rho)\subseteq
\dom(p,\Rset ^d\setminus C)$, and so it remains to prove
$C\subseteq D$.

\labepsfig{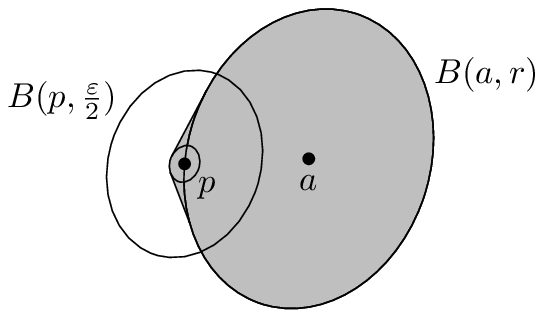}{The sets $C$ (shaded) and $D$.}

To this end, it is sufficient to prove the following:
If $B := \cBall (0, 1)$ is the unit ball of $\| \mathord\cdot \|$ and $\eta>0$
is given, then there exists $\delta>0$ such that for
every $x\in\Rset ^d$ with $\|x\| \le 1 + \delta$, the
``cap'' $\conv(B\cup\{x\})\setminus B$ has diameter
at most $\eta$. This is a well-known and easily proved
property of uniformly convex norms. (Proof sketch:
If $x$ with $\|x\|=1+\delta$ has a cap of large diameter,
then there is $z$ of norm $1$ and half of the diameter
away from $x$
such that the line $xz$ avoids
the interior of $B$. Let $y$ be the other intersection of
this line with $\bd \cBall (0,1+\delta)$---then $xy$ is a long
segment that cuts in $\cBall (0,1+\delta)$ into depth only $\delta$.)
\end{proof}

\begin{proof}[Proof of Theorem~\ref{t:norms}] The overall strategy
of the proof is exactly as for Theorem~\ref{t:euclcase}
(see Section~\ref{s:euclcase}).
The constant $\alpha$ in \ref{enumi: length decreases}
may also depend on the considered norm $\| \mathord\cdot \|$. 
This quantification also needs to be added in 
the appropriate version of Lemma~\ref{l:step}.

In the proof of that lemma, the first place where we use a property
not shared by all norms is below (\ref{e:delta'}); we need that
the triangle inequality may hold with equality only for collinear
points---this remains true for all \emph{rotund} norms.

Then we proceed as in the Euclidean case, introducing the
the cone $K=\conv(\{a\}\cup B(p,\rho))$ as in
Lemma~\ref{lemma:cone}. There is some $\gamma>0$,
depending on $\eps$, $s_0$, and the norm $\| \mathord\cdot \|$, 
such that the appropriate Euclidean cone with opening angle
$\gamma$ is contained in $K$. (Here and in the sequel
we implicitly use the fact that every norm on $\R^d$ is
between two constant multiples of the Euclidean norm,
which is well known and immediate by compactness.) 

We define the rays $k$ and $\ell$, again following the Euclidean proof.
For the next step, we need that, since the angle of these rays
is bounded away from the straight angle, at least one of
$k,\ell$ cuts a segment of a significant length $\beta$
from the ball $\cBall(b',r)$. It is easy to see that this
property follows from the \emph{smoothness} of the norm.
The rest of the proof goes through unchanged.
\end{proof}

\section{Non-uniqueness examples}\label{s:nonuniq}

As we saw in the introduction, two point sites with the
same $x$-coordinate have at least two zone diagrams under
the $\ell_1$ metric. Here we show that only the non-smoothness
(sharp corners) of the $\ell_1$ unit ball is essential for this
example, while the straight edges can be replaced by curved ones.

\begin{prop}\label{p:}
 There exists a rotund norm in the plane, arbitrarily close
to the $\ell_1$ norm, such that two distinct point sites
with the same $x$-coordinate have (at least) two different zone diagrams.
\end{prop}

The appropriate norm is not difficult to describe, but proving the
non-uniqueness of the zone diagram is more demanding, since
it seems hard to find an explicit description of a zone diagram
for non-polygonal norms.

Informally, we construct the desired norm by slightly ``inflating''
the unit ball of the planar $\ell_1$ norm, so that the edges
bulge out and the norm becomes rotund. It is important
that the inflation is asymmetric, as is schematically
indicated in Fig.~\ref{f:bulgefat} (in the ``real'' example
we inflate much less). 
\labepsfig{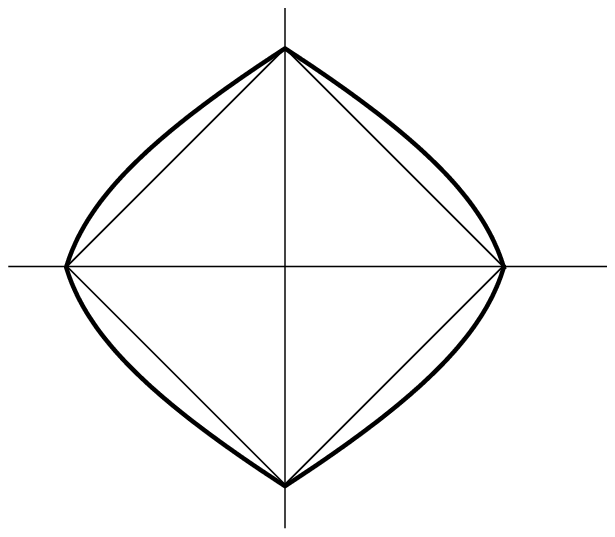}{A schematic illustration of the unit ball of $\| \mathord\cdot \|_{(1)}$.}
We will denote the resulting norm by $\| \mathord\cdot \|_{(1)}$; 
the subscript should remind of ``inflated $\ell_1$'' graphically.

\labepsfig{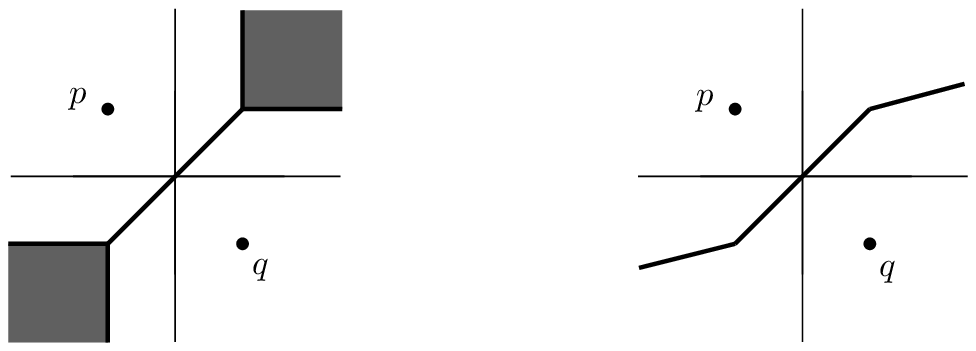}{The bisector of $p$ and $q$ under the $\ell_1$ norm
and under $\| \mathord\cdot \|_{(1)}$ (schematic).}

To explain the purpose of the asymmetry in our example,
we consider the bisector of the points $p=(-1,1)$ and $q=(1,-1)$,
i.e., the set of all points equidistant to $p$ and $q$.
For the $\ell_1$ norm, the bisector is ``fat'', as shown in
Fig.~\ref{f:bisects} left---it consists of a segment and
two quadrants. By a small inflation, which makes the norm
rotund, the middle segment of
the bisector is changed only very slightly, but the 
``ambiguity'' of the $\ell_1$ bisector in the quadrants
is ``resolved'', and the quadrants collapse to
(possibly curved) rays.
Now if the inflation were symmetric, 
we would get
straight rays with slope $1$ in the bisector, but with an asymmetric inflation,
we can get a (positive) slope as small as we wish.

In order to establish the required properties of the bisector
formally, a safe route (if perhaps not the most conceptual one)
is to describe $\| \mathord\cdot \|_{(1)}$ analytically. The rays
of the bisector will be slightly curved rather than straight, but 
for the zone diagram construction this will do as well. 

\begin{lemma}\label{l:(1)norm}
For every $\eps>0$ there exists a rotund norm $\| \mathord\cdot \|_{(1)}$ in the plane,
whose unit ball contains the $\ell_1$ unit ball and is contained
in the octagon as in Fig.~\ref{f:outeroctagon} left, such that the
portion of the bisector of the points $p=(-1,1)$ and $q=(1,-1)$
lying in the quadrant $\{\, (x,y) : x, y \geq 1 \,\}$ is an $x$-monotone
curve lying below the line $y=\eps(x-1)+1$ (Fig.~\ref{f:outeroctagon} right).
\labepsfig{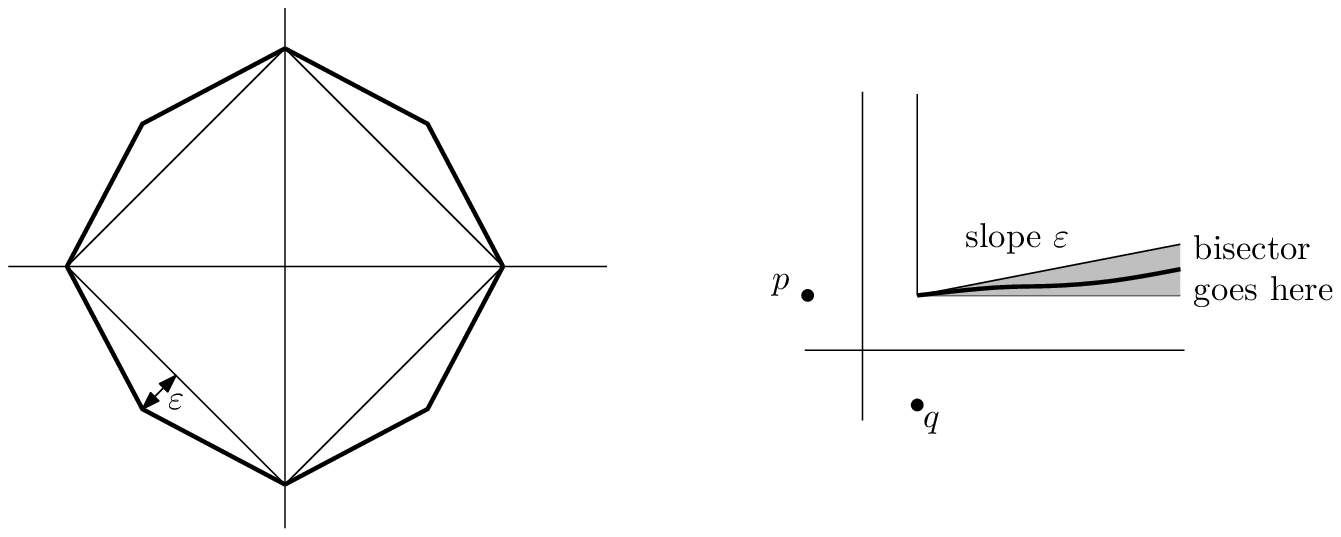}{The conditions in Lemma~\ref{l:(1)norm}. }
\end{lemma}

See Appendix~\ref{appendix section: inflated norm lemma} for a proof. 

\begin{proof}[Proof of Proposition~\ref{p:}]
We show that the zone diagram of the sites $p^-=(0,-1)$
and $p^+=(0,+1)$ under the norm $\| \mathord\cdot \|_{(1)}$ as 
in the lemma, with $\eps$ sufficiently small, is not unique.

\looseness=-1 
First we consider the zone diagram only inside the vertical strip
\begin{equation*}
V:=\{\, (x,y)\in\R^2:x\in [-2, 2] \,\}.
\end{equation*}
Let $R_0^+$ be the region
as in 
Fig.~\ref{f:twozono}, 
i.e., the part of the region
of $p^-$ within $V$ in an $\ell_1$ zone diagram of $p^-,p^+$. 
Let $S_0^+$ be obtained by pulling the bottom vertex of $R_0^+$
downward by $\eta$ (which is another small positive parameter),
 and let $R_0^-,S_0^-$ be the reflections
of $R_0^+,S_0^+$ by the $x$-axis.

\labepsfig{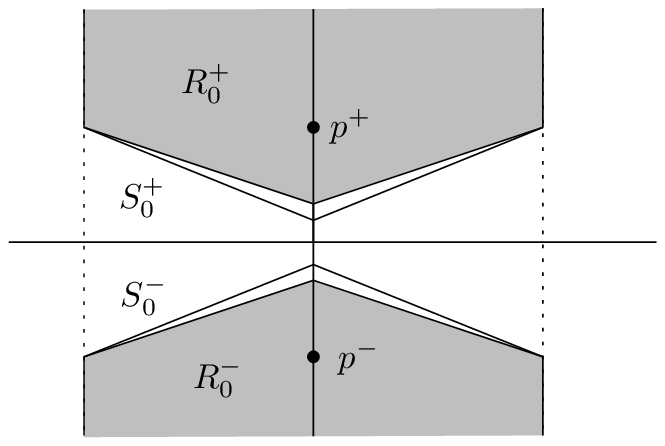}{The regions $R_0^+,S_0^+,R_0^-,S_0^-$ in the vertical strip~$V$.}

Let us consider the region $\dom(p^-,R_0^+)$ inside $V$ (distances
measured by our norm $\| \mathord\cdot \|_{(1)}$). For every point $x\in V$ below $R_0^+$,
the $\| \mathord\cdot \|_{(1)}$-distance to $R_0^+$ coincides with the $\ell_1$
distance, which is simply the length of the vertical segment from $x$
to $\bd R_0^+$. From this it is clear that 
$\dom(p^-,R_0^+)\supseteq R_0^-$ (since $R_0^-$ is the dominance
region of $p^-$ against $R_0^+$ in the $\ell_1$ metric,
and $\| \mathord\cdot \|_{(1)}\le \| \mathord\cdot \|_1$). Moreover, it's easy to check
that for $\eps$ (the parameter controlling the
choice of $\| \mathord\cdot \|_{(1)}$) sufficiently small, we also have
$\dom(p^-,R_0^+)\subseteq S_0^-$.

Thus, we have $R_0^-\subseteq \dom(p^-,R_0^+)\subseteq S_0^-$,
and by the vertical symmetry we also get 
$R_0^+\subseteq \dom(p^+,R_0^-)\subseteq S_0^+$.
Arguing as in either of the proofs of Theorem~\ref{t:doublezone}, 
we get that there exist regions $R^+,R^-,S^+,S^-$,
where $R^-$ is the reflection of $R^+$, $S^-$ is the reflection
of $S^+$, such that $R_0^-\subseteq R^+\subseteq S^+\subseteq S^+$,
and $(R^-,S^+)$ is a zone diagram of $(p^-,p^+)$
(and so is $(S^-,R^+)$, but we actually have
$R^+=S^+$, although we will neither need this nor prove it).

All of this refers to the vertical strip $V$ (so, formally,
the metric space in these arguments is $V$ with the $\| \mathord\cdot \|_{(1)}$
metric). Now we move on to the full plane $\R^2$, and we let $\tilde S^+$
be the region consisting of $S^+$ plus two parts of the upper
halfplane outside $V$ as in 
Fig.~\ref{f:exttendRS}: 
The right part is delimited by a part of the bisector of
$p^+$ and $(2,-1)$ (drawn thick), and the left part
by a part of the  bisector of $p^+$ and $(-2,-1)$.

Now we set $\tilde R^-:=\dom(p^-,\tilde S^+)$.
The distance of points inside $V\setminus S^+$ to $\tilde S^+$
is still the vertical distance, i.e., the same as
the distance to $S^+$, and so $\tilde R^-\cap V=R^-$.
For the part of $\tilde R^-$ outside $V$, we don't need
an exact description---it is sufficient that it lies below
the dashed rays in 
Fig.~\ref{f:exttendRS}
(using the property of the bisectors as
in Lemma~\ref{l:(1)norm}, one can see that these rays
can be taken as steep as desired, by setting $\eps$
sufficiently small). From this we can see that
for every point of the upper halfplane on the right of $V$,
the nearest point of $\tilde R^-$ is the corner $(2,-1)$.

\labepsfig{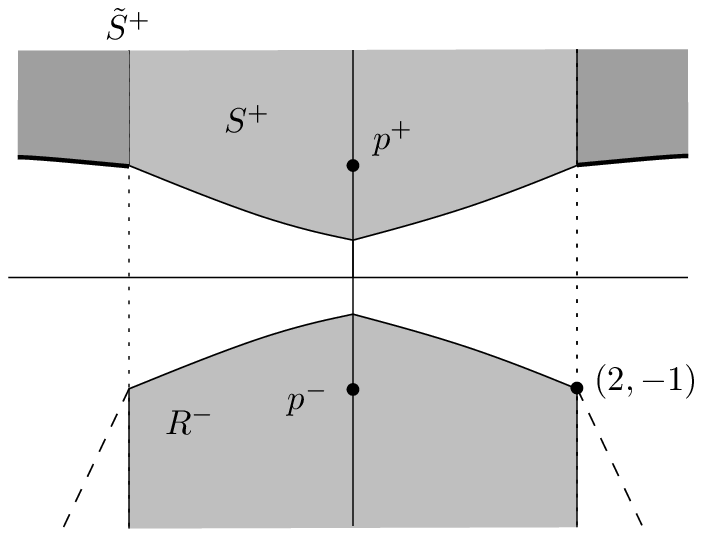}{The region $\tilde S^+$ defined using bisectors,
and a region containing~$\tilde R^-$.}

Therefore, $\dom(p^+,\tilde R^-)=\tilde S^+$, and
hence $(\tilde R^-,\tilde S^+)$ is a zone diagram 
of $(p^-,p^+)$. But the mirror reflection of this zone
diagram about the $x$-axis yields another, different zone diagram.
\end{proof}

\subsection*{Acknowledgements}
We are grateful to Tetsuo Asano for valuable discussions 
including those on non-uniqueness examples for convex polygonal distances.
We also express our gratitude to Daniel Reem for 
careful reading and useful suggestions on the manuscript. 
Finally, we remark that the warm comments from the audience 
of our preliminary announcement of partial results at EuroCG~2009 
encouraged us to work further. 

\makeatletter\@topnum\z@\makeatother

\appendix

\section{Proof of Theorem~\ref{theorem: smooth two singleton plane}}
\label{s:smoothonly}

Proposition~\ref{p:} showed that 
the assumption of smoothness in Theorem~\ref{t:norms} cannot be dropped, 
even for the simplest case of two singleton sites in the plane. 
Theorem~\ref{theorem: smooth two singleton plane}, 
which we will prove here, 
states that the rotundity assumption can be dropped in this special case. 

Smoothness of the norm means that 
a metric ball has a unique supporting halfspace at every point in its surface. 
Thus, for a nonzero vector~$a$, 
we can define $\support ^{> 0} _a$ to be
the open halfspace
that touches (but not intersects)
the ball $\cBall (-a, \lVert a \rVert)$ at the origin. 
We write $\support ^{\leq 0} _a = \Rset ^d \setminus \support ^{> 0} _a$ 
and $\support ^{\geq 0} _a = \support ^{\leq 0} _{-a}$. 
\begin{figure}
\begin{center}
\includegraphics[clip,scale=1.1]{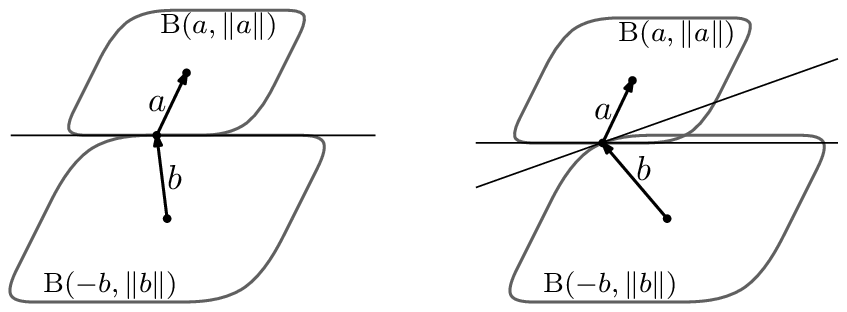}
\caption{$\lVert a + b \rVert = \lVert a \rVert + \lVert b \rVert$ if and only if
         $a \sim b$ 
         (equation \eqref{equation: equivalence} with $m = 2$).}
\label{figure: balls intersect}
\end{center}
\end{figure}
For nonzero vectors $a$ and $b$, 
define $a \sim b$ when $
\support ^{> 0} _a = \support ^{> 0} _b
$.  Then $\sim$ is an equivalence relation. 
It is easy to see (Fig.~\ref{figure: balls intersect}) 
that for nonzero vectors $a _1$, \ldots, $a _m$, we have
\begin{equation}
\label{equation: equivalence}
\lVert a _1 + \dots + a _m \rVert = \lVert a _1 \rVert + \dots + \lVert a _m \rVert
\qquad
\text{if and only if}
\qquad
a _1 \sim \dots \sim a _m. 
\end{equation}

\begin{lemma}
\label{lemma: almost straight}
Let $\lVert \mathord\cdot \rVert$ be a smooth norm on $\Rset ^d$. 
Then there are positive numbers $\alpha$ and $\beta$ such that 
for any unit vectors $u$, $v$ with 
$\lVert u + v \rVert > 2 - \beta$, 
we have $\lVert u - \alpha v \rVert \leq 1$. 
\end{lemma}

\begin{proof}
The angle~$\sigma _u$ between a unit vector~$u$ and 
$\support ^{\leq 0} _u$ is 
a continuous function of $u$, 
and hence attains a positive minimum~$\sigma$. 
Let $\support ^{\geq \sigma / 2} _u$ (and $\support ^{\leq \sigma / 2} _u$) be 
the set of vectors (including $0$) 
that make an angle $\geq \sigma / 2$ (and $\leq \sigma / 2$)
with $\support ^{\leq 0} _u$ (Fig.~\ref{figure: gamma-cone}). 
\begin{figure}
\begin{center}
\includegraphics[clip,scale=1.2]{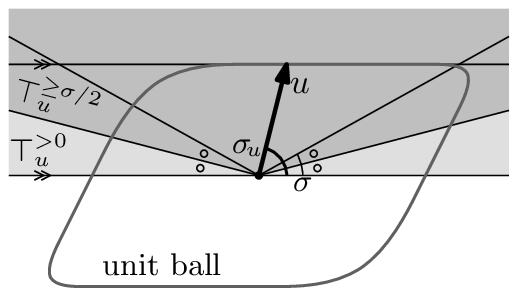}
\caption{$\support ^{\geq \sigma / 2} _u$ is the set of vectors that are significantly 
         closer to $u$ than to $-u$.}
\label{figure: gamma-cone}
\end{center}
\end{figure}
We find the desired $\alpha$ and $\beta$ as follows. 

For unit vectors $u$ and $v$ with $v \in \support ^{\geq \sigma / 2} _u$, 
let $\alpha _{u, v}$ be the length of the segment that the unit ball 
cuts out from the line $u + \Rset v$. 
In other words, $\alpha _{u, v}$ is the unique positive number such that 
$\lVert u - \alpha _{u, v} v \rVert = 1$. 
Then $\alpha _{u, v}$ is continuous in $u$ and $v$, 
and thus attains a positive minimum $\alpha$. 

For unit vectors $u$ and $v$ with $v \in \support ^{\leq \sigma / 2} _u$, 
let $\beta _{u, v} = 2 - \lVert u + v \rVert$. 
Then $\beta _{u, v}$ is positive and continuous in $u$ and $v$, 
and thus attains a positive minimum $\beta$. 

Since $\support ^{\geq \sigma / 2}$ and $\support ^{\leq \sigma / 2}$
covers the whole space, 
$\alpha$ and $\beta$ have the stated property. 
\end{proof}

\begin{lemma}
\label{lemma: balls}
Let $\lVert \mathord\cdot \rVert$ be a smooth norm on $\Rset ^d$. 
For any $\kappa > 0$, 
there is $\varepsilon > 0$ such that, 
for any vectors $u$, $v$ with 
$\lVert u \rVert$, $\lVert v \rVert \geq 1$ and 
$\lVert u - v \rVert < \varepsilon$, 
we have $\dist (y, \cBall (u, \lVert u \rVert)) < \kappa \lVert y \rVert$ 
for any $y \in \cBall (v, \lVert v \rVert)$. 
\end{lemma}

\begin{proof}
Since $\dist (y, \cBall (u, \lVert u \rVert)) \leq 2 \varepsilon$, 
it is clear that, for any constant $\eta > 0$, 
the claim holds if we consider only those $y$ with $\lVert y \rVert \geq \eta$. 
Therefore, it suffices to prove the existence of $\eta > 0$, 
depending on $\lVert \mathord\cdot \rVert$ and $\kappa$, 
such that the claim holds for any $y$ with $\lVert y \rVert < \eta$. 

We find the desired $\eta$ and $\varepsilon$ as follows
(Fig.~\ref{figure: balls shifted}). 
Since the norm is smooth, the surface of a ball 
looks like a hyperplane locally at each point. 
Thus, there exists $\eta > 0$ such that 
for any $u \in \Rset ^d$ with $\lVert u \rVert \geq 1$
and any $z \in \support ^{\geq 0} _u$ with $\lVert z \rVert < \eta (1 + \kappa / 2)$, 
we have $
\dist (z, \cBall (u, \lVert u \rVert)) \leq \kappa \lVert z \rVert / (2 + \kappa)
$.  
Also, since changing slightly  a vector~$u$ of length $1$ or greater 
moves $\support ^{\geq 0}$ only slightly, 
there is $\varepsilon > 0$ so small that 
for any vectors $u$, $v$ of length $1$ or greater with
$\lVert u - v \rVert < \varepsilon$, 
we have $\dist (y, \support ^{\geq 0} _u) \leq \kappa \lVert y \rVert / (2 \eta)$
for all $y \in \support ^{\geq 0} _v$. 

\begin{figure}
\begin{center}
\includegraphics[clip,scale=1.1]{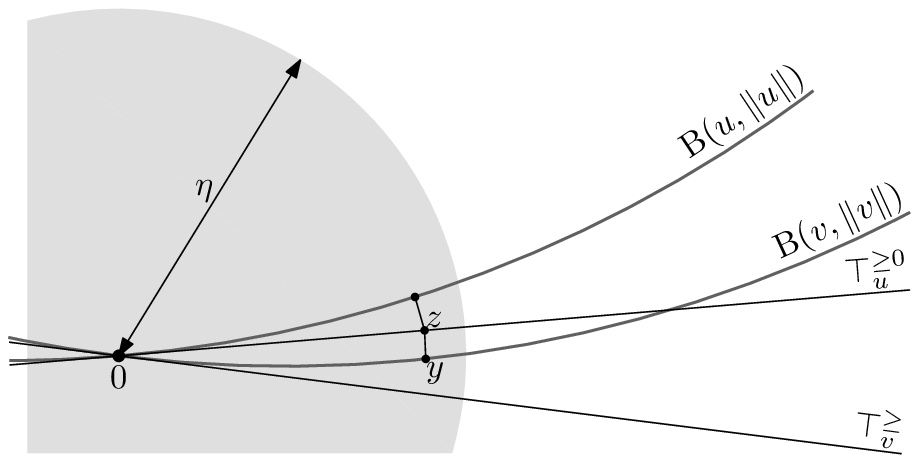}
\caption{When $u$ and $v$ are close, $y \in \cBall (v, \lVert v \rVert)$ 
         is not very far from $\cBall (u, \lVert u \rVert)$.}
\label{figure: balls shifted}
\end{center}
\end{figure}

Since $y \in \cBall (v, \lVert v \rVert) \subseteq \support ^{\geq 0} _v$, 
we have $\dist (y, \support ^{\geq 0} _u) \leq \kappa \lVert y \rVert / 2$ 
by our choice of $\varepsilon$. 
Let $z \in \support ^{\geq 0} _u$ be a point attaining this distance. 
Since $
 \lVert z \rVert 
\leq
 \lVert y \rVert + \lVert z - y \rVert 
\leq
 \lVert y \rVert + \kappa \lVert y \rVert / 2
=
 \lVert y \rVert (1 + \kappa / 2)
\leq 
 \eta (1 + \kappa / 2)
$, we have $
 \dist (z, \cBall (u, \lVert u \rVert)) 
\leq
 \kappa \lVert z \rVert / (2 + \kappa) 
\leq
 \kappa \lVert y \rVert / 2
$ by our choice of $\eta$. 
These imply $\dist (y, \cBall (u, \lVert u \rVert)) < \kappa \lVert y \rVert$ 
by the triangle inequality. 
\end{proof}

\begin{lemma}
\label{lemma: separate}
Let $\lVert \mathord\cdot \rVert$ be a smooth norm on $\Rset ^2$. 
For unit vectors $u$ and $v$ with $\lVert u - v \rVert < 2$, 
there is $\kappa > 0$ such that 
for all $
y \in \dom (v, u) \setminus \cBall (v, 1) 
$ sufficiently close to the origin (Fig.~\ref{figure: balls_tangent}), 
$\dist (y, \cBall (u, 1)) \geq \kappa \lVert y \rVert$. 
\begin{figure}
\begin{center}
\includegraphics[bb = 10 670 170 780,clip,scale=1.1]{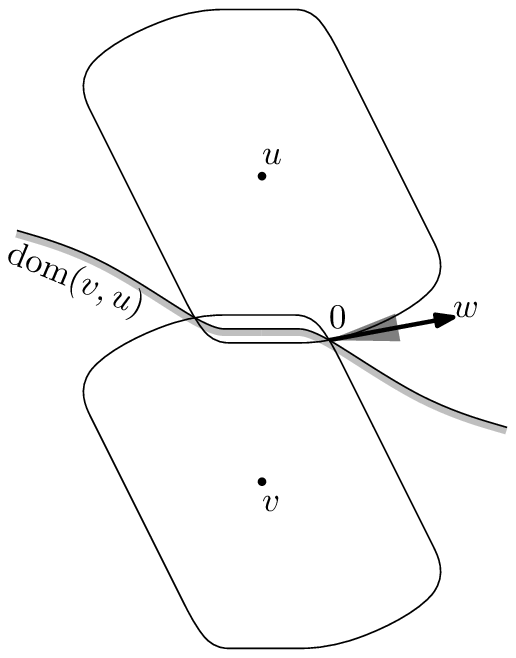}
\caption{The conclusion of Lemma~\ref{lemma: separate} states that 
         $\dom (v, u)$ and the boundary of $\cBall (u, 1)$
         ``make a positive angle'' at the origin.
         We prove this by showing that there is a cone (shaded) 
         whose axis is the tangent vector~$w$
         and which does not overlap $\dom (v, u)$.}
 \label{figure: balls_tangent}
\end{center}
\end{figure}
\end{lemma}

\begin{proof}
Because $\lVert u - v \rVert < 2$, 
the vectors $u$ and $-v$ do not share the tangent. 
Therefore, there is a (unique) unit vector $
w \in \support ^{\geq 0} _u \cap \support ^{\leq 0} _u
$ that heads out of $\cBall (v, 1)$. 
Since 
\begin{align*}
 \lim _{\delta \to 0} \frac{\lVert u - \delta w \rVert - 1}{\delta} & = 0, 
&
 \beta := \lim _{\delta \to 0} \frac{\lVert v - \delta w \rVert - 1}{\delta} & > 0, 
\end{align*}
there exists $\delta _0 > 0$ so small that 
for all positive $\delta < \delta _0$, 
we have 
\begin{align*}
 \frac{\lVert u - \delta w \rVert - 1}{\delta} & < \frac 1 3 \beta, 
&
 \frac{\lVert v - \delta w \rVert - 1}{\delta} & > \frac 2 3 \beta, 
\end{align*}
and hence $
 \lVert u - \delta w \rVert < \lVert v - \delta w \rVert - \beta \delta / 3
$.  This implies that 
$\lVert u - x \rVert < \lVert v - x \rVert$ 
for all $x \in \cBall (\delta w, \beta \delta / 6)$. 
Thus, $\dom (v, u)$ is disjoint from 
a cone (except at the origin)
whose vertex is at the origin and axis is the vector~$w$ 
(see Fig.~\ref{figure: balls_tangent}). 
This implies what is stated. 
\end{proof}

Now we look at the situation of Theorem~\ref{theorem: smooth two singleton plane}. 
Let $\RR = (R _0, R _1)$ and $\SS = (S _0, S _1)$ be pairs 
satisfying $\RR \preceq \SS$ and $\RR = \Dom \SS$, $\SS = \Dom \RR$
(which exist by Theorem~\ref{t:doublezone}). 
As before, it suffices to show that $\RR = \SS$. 
Suppose otherwise. 
Then $
 h
= 
 \min \{\dist (p _0, S _0 \setminus R _0), \dist (p _1, S _1 \setminus R _1)\}
$ exists. 

\begin{lemma}
 \label{lemma: 2 l}
In the above setting, 
if a point $c \in \closure{S _0 \setminus R _0}$
satisfies $\lVert c - p _0 \rVert = h$, 
then 
\begin{enumerate}
\renewcommand{\theenumi}{\textup{(\alph{enumi})}}
\renewcommand{\labelenumi}{\theenumi}
\item \label{enumi: 2 l}
$\lVert c - p _1 \rVert = 2 h$; 
\item \label{enumi: conjugate}
there is a point $c' \in \closure{S _1 \setminus R _1}$
satisfying $\lVert c' - c \rVert = \lVert c' - p _1 \rVert = h$. 
\end{enumerate}
\end{lemma}

\begin{proof}
Note that $c \in R _0$, since 
otherwise $S _0 \setminus R _0$ intersects 
a part of the segment $c p _0$ of positive length, 
contradicting the minimality of $h$. 

There is a sequence $(x _i) _{i \in \Nset}$ of points in $S _0 \setminus R _0$ 
that converges to $c$. 
For each $i \in \Nset$, 
let $y _i \in S _1$ be a closest point to $x _i$. 
Since $x _i \in S _0 \setminus R _0$, 
we have $\lVert y _i - x _i \rVert = \dist (x _i, S _1) < \lVert p _0 - x _i \rVert$ 
and $y _i \in S _1 \setminus R _1$. 
The sequence $(y _i) _{i \in \Nset}$ has a subsequence $(y _{j _i}) _{i \in \Nset}$ 
that converges to a point $c' \in \closure{S _1 \setminus R _1}$
(Fig.~\ref{figure: four balls}).
\begin{figure}
\begin{center}
\includegraphics[clip,scale=1.1]{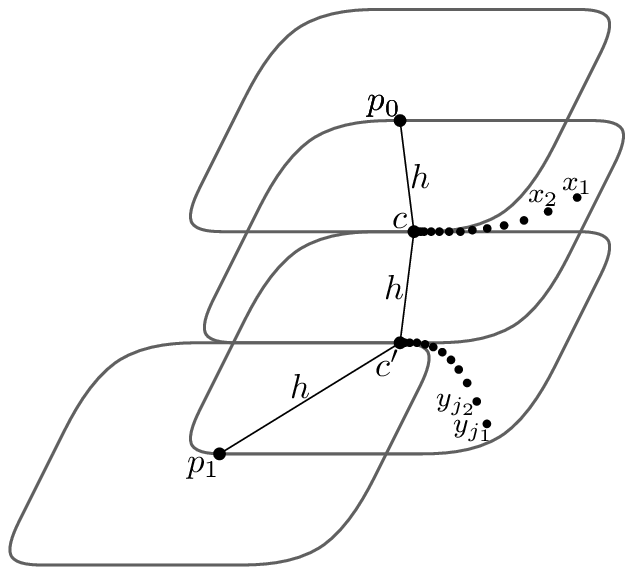}
\caption{Lemma~\ref{lemma: 2 l}.}
 \label{figure: four balls}
\end{center}
\end{figure}
Note that 
\begin{equation*}
\label{equation: c _0 and c _1}
 \lVert c' - p _1 \rVert
\leq 
 \lVert c - c' \rVert 
=
 \lim _{i \to \infty} \lVert x _{j _i} - y _{j _i} \rVert 
\leq 
 \lim _{i \to \infty} \lVert p _0 - x _{j _i} \rVert 
=
 \lVert p _0 - c \rVert
=
 h, 
\end{equation*}
where the first inequality is by $c' \in S _1$ and $c \in R _0$. 
In fact, this
holds in equality by the minimality of $h$. 
We have proved \ref{enumi: conjugate}. 

For each $i$, 
since $S _1 \setminus R _1$ intersects a part of the segment $y _{j _i} c'$
of positive length, 
$y _{j _i} \notin \cBall (p _1, h)$ by the minimality of $h$. 
Also, $
 y _{j _i} 
\in 
 S _1 
\subseteq
 \dom (p _1, c) 
$. 
As $i$ increases, $y _{j _i}$ comes arbitrarily close to $c'$. 
Hence, 
if \ref{enumi: 2 l} is not true, 
Lemma~\ref{lemma: separate} gives a constant $\kappa > 0$ such that 
$\dist (y _{j _i}, \cBall (c, h)) \geq \kappa \lVert y _{j _i} - c' \rVert$
for all but finitely many $i$. 
On the other hand, since 
$y _{j _i}$ is in $\cBall (x _{j _i}, \lVert x _{j _i} - c' \rVert)$ and 
$(x _{j _i}) _{i \in \Nset}$ converges to $c$, 
Lemma~\ref{lemma: balls} shows that $
\dist (y _{j _i}, \cBall (c, h)) < \kappa \lVert y _{j _i} - c' \rVert
$ for all but finitely many $i$. 
This is a contradiction.  We have proved \ref{enumi: 2 l}. 
\end{proof}

\begin{lemma}
 \label{lemma: 3 l}
In the above setting, 
$\lVert p _0 - p _1 \rVert = 3 h$. 
\end{lemma}

\begin{proof} 
By the definition of $h$, 
there is a point $c \in \closure{S _0 \setminus R _0}$ 
satisfying $\lVert c - p _0 \rVert = h$. 
By Lemma~\ref{lemma: 2 l}\ref{enumi: conjugate}, 
there is a point $c' \in \closure{S _1 \setminus R _1}$ 
satisfying $\lVert c' - c \rVert = \lVert c' - p _1 \rVert = h$. 
By Lemma~\ref{lemma: 2 l}\ref{enumi: 2 l} 
(and the same lemma with the sites swapped), 
$\lVert c - p _1 \rVert = \lVert c' - p _0 \rVert = 2 h$. 
This implies $(c - p _0) \sim (c' - c) \sim (p _1 - c')$
and thus $\lVert p _0 - p _1 \rVert = 3 h$ 
by \eqref{equation: equivalence} at the beginning of this section. 
\end{proof}

To prove Theorem~\ref{theorem: smooth two singleton plane}, 
we will construct a sequence $(b _t) _{t \in \Nset}$ of points 
in $R \setminus S$, as we did in Section~\ref{s:euclcase}. 
Recall that for each $i \in \{0, 1\}$ and $b \in S _i$, 
we define $a (b)$ to be the 
closest point to $b$ 
that is in the intersection of $R _i$ with the segment $b p _i$
(note that since we do not have the cone lemma this time, 
the intersection of $b p _i$ and $\bd R _i$ is not always unique). 
As before, let $s (b) = \lVert b - p _i \rVert$ 
and $\delta (b) = \lVert b - a (b) \rVert$. 

The proof goes as follows. 
This time, we begin with a point~$b _0 \in S _0 \setminus R _0$ that is 
within distance $h + \varepsilon$ from the nearest site, 
for some small $\varepsilon > 0$
(such $b _0$ exists by the definition of $h$), 
and take $b _1$, $b _2$, \ldots 
as we did in Section~\ref{s:euclcase} using Lemma~\ref{l:step}: 
For each $b _t \in S _i \setminus R _i$, 
we let $b _{t + 1} \in S _{1 - i} \setminus R _{1 - i}$ be a point 
that is at the same distance from $a (b _t)$ as $p _i$ is. 
Then each $b _t$ will be also 
within distance $h + \varepsilon$ from the nearest site~$p _i$. 
Because we have proved that the sites are $3 h$ apart, 
and the path $p _i$-$a (b _t)$-$b _{t + 1}$-$p _{1 - i}$ consists of 
three segments shorter than $h + \varepsilon$, 
the path must be ``almost straight''. 
This implies that we will always have 
the case~\ref{enumi: gap increases} in Section~\ref{s:euclcase}
(Fig.~\ref{f:cutball} left): 

\begin{lemma}
 \label{lemma: gap increases (smooth)}
In the above setting, 
the following holds for some $\varepsilon > 0$:
For each $i \in \{0, 1\}$ and $
b \in S _i \setminus R _i 
$ satisfying $s := s (b) < h + \varepsilon$, there is $
b' \in S _{1 - i} \setminus R _{1 - i} 
$ such that $\delta := \delta (b)$, $s' := s (b')$, $\delta' := \delta (b')$ satisfy
\ref{enumi: gap increases} of Section~\ref{s:euclcase} 
(i.e., $\delta' \geq \delta$ and $s' \leq s - \delta$). 
\end{lemma}

\begin{proof}
Let $\varepsilon := \min \{h \alpha, h \beta / 3\}$, 
where $\alpha$ and $\beta$ are as in Lemma~\ref{lemma: almost straight}. 
Let $b$ be as assumed. 
By the definition of $a := a (b)$, 
there is $b' \in S _{1 - i}$ 
with $\lVert b' - a \rVert = \lVert a  - p _i \rVert$. 
We show that this $b'$ qualifies. 
Since $
 s' 
=
 \lVert b' - p _{1 - i} \rVert 
\leq
 \lVert b' - a \rVert
=
 \lVert a - p _{1 - i} \rVert
=
 s - \delta
$, 
it suffices to prove that $\delta' \geq \delta$
(which would then imply $b' \notin R _{1 - i}$). 

By Lemma~\ref{lemma: 3 l}, we have 
\begin{align*}
  \lVert b' - p _i \rVert 
&
 \geq
  \lVert p _{1 - i} - p _i \rVert - \lVert p _{1 - i} - b' \rVert 
 =
  3 h - s'
 >
  3 h - s
 \geq
  3 h - (h + \varepsilon) 
\notag
\\
&
 =
  2 (h + \varepsilon) - 3 \varepsilon
 \geq
  2 (h + \varepsilon) - \beta h
 >
  (h + \varepsilon) (2 - \beta)
 > 
  \lVert a - p _i \rVert (2 - \beta). 
\end{align*}
By this and $\lVert b' - a \rVert = \lVert a - p _i \rVert$, 
Lemma~\ref{lemma: almost straight} yields
$\lVert (b' - a) - \alpha (a - p _i) \rVert \leq \lVert a - p _i \rVert$. 
This remains true if we decrease $\alpha$, 
since $\cBall (0, \lVert a - p _i \rVert)$ is convex. 
So we replace $\alpha$ by $
 \lVert b - a \rVert / \lVert a - p _i \rVert
\leq
 \varepsilon / h
\leq
 \alpha
$, obtaining $
  \lVert b' - b \rVert
 =
  \lVert (b' - a) - (b - a) \rVert
 \leq
  \lVert a - p _i \rVert
$. 

Since $b$ is in $S _i$
and $a' := a (b')$ is in $R _{1 - i}$, 
we have $\lVert a' - b \rVert \geq s$. 
Hence, $
  \delta' 
 =
  \lVert b' - a' \rVert
 \geq
  \lVert a' - b \rVert - \lVert b' - b \rVert
 \geq
  s - \lVert a - p _i \rVert
 =
  \delta
$, as desired. 
\end{proof}

The rest of the argument is similar to what we already saw in Section~\ref{s:euclcase}
(and even simpler because we do not have case~\ref{enumi: length decreases} this time):
Starting at $b _0 \in S \setminus R$ such that $s (b _0) < h + \varepsilon$, 
where $\varepsilon$ is as in Lemma~\ref{lemma: gap increases (smooth)}, 
we define $b _{t + 1}$, 
for each $t \in \Nset$, 
to be the point $b'$ corresponding to $b = b _t$. 
By the lemma, 
$s (b _t)$ always decreases by at least $\delta (b _0)$, 
leading to a contradiction. 
This proves Theorem~\ref{theorem: smooth two singleton plane}. 

\section{Proofs of Theorem~\ref{t:doublezone}}
\label{appendix section: double zone diagrams}

There are two proofs of Theorem~\ref{t:doublezone}
available; we sketch the main ideas for the reader's
convenience.

The \emph{first proof}, from \cite{asano07:_zone_diagr},
doesn't establish the theorem in full generality---it
works only for closed and disjoint
sites in a Euclidean space, or more generally,
in a finite-dimensional normed space with a rotund norm.
In this proof, 
we build a sequence of inner approximations to $\RR$ and outer
approximations to $\SS$. Namely, we set $\RR^{(0)}:=\PP$,
$\SS^{(0)}:=\Dom \RR^{(0)}$ (this is the classical Voronoi diagram
of the sites $P_1,\ldots,P_n$), and for $k=1,2,\ldots$ we put
$\RR^{(k)}:=\Dom\SS^{(k-1)}$, $\SS^{(k)}:=\Dom\RR^{(k-1)}$.

Antimonotonicity of $\Dom$ and induction yield
$\RR^{(0)}\preceq \RR^{(1)}\preceq \RR^{(2)}\preceq\cdots$
and $\SS^{(0)}\succeq \SS^{(1)}\succeq \SS^{(2)}\succeq\cdots$,
as well as $\RR^{(k)}\preceq \SS^{(k)}$ for all $k$.
We then define $\RR$ and $\SS$ by
\begin{align*}
R_i & := \closure{\bigcup_{k=0}^\infty R_i^{(k)}}, &
S_i & := \bigcap_{k=0}^\infty S_i^{(k)}.
\end{align*}
It remains to show that $\RR$ and $\SS$ are as required.
This is done in \cite{asano07:_zone_diagr} for the case
of point sites in $\R^2$ with the Euclidean norm.
By inspecting the proof (Lemma~5.1 of \cite{asano07:_zone_diagr}),
we see that it uses only the following property of the underlying metric
space (stated there as Lemma~3.1):  \emph{If $P$ is a closed
set, $X_1\supseteq X_2\supseteq \cdots$ is a decreasing sequence
of closed sets with $X_1\cap P=\emptyset$, 
and $X:=\bigcap_{k=1}^\infty X_k$, then $\dom(P,X)
\subseteq \closure{\bigcup_{k=1}^\infty \dom(P,X_k)}$.}
(Moreover, in the proof one also needs that $P_i\cap S^{(0)}_j=\emptyset$
for $i\ne j$; since we assume the sites to be closed and 
disjoint, this property of the Voronoi diagram is immediate.)

To verify the above statement, we can again follow the proof
of Lemma~3.1 in \cite{asano07:_zone_diagr}. First we check
that with the $X_k$ as above and any point $y$, 
we have $\dist(y,X)=\lim_{k\to\infty} \dist(y,X_k)$;
this follows easily assuming compactness of all closed balls 
in a finite-dimensional normed space. Now let us fix $x\in\dom(P,X)$
arbitrarily (we may assume $x\not\in P$, since
the case $x\in P$ is clear) and choose $\eps>0$; we want to show that
$\dist(x,\dom(P,X_k))\le\eps$ for some $k$. We let $p$ be
a point of $P$ nearest to $x$, and choose a point
$y\ne x$ on the segment $px$ at distance smaller than $\eps$ from~$x$.
It is easy to check, using the rotundity of the norm, that
$\dist(y,p)<\dist(y,X)$, and thus $\dist(y,p)\le\dist(y,X_k)$
for $k$ sufficiently large. So $y\in\dom(P,X_k)$ and we are done.

\medskip

The \emph{second proof} of Theorem~\ref{t:doublezone}, due to Reem and Reich
\cite{reem07:_zone_and_doubl_zone_diagr}, is based on the following
theorem of Knaster and Tarski (see
\cite{tarski55:_lattic_theor_fixpoin_theor_and_its_applic}):
{\sl If $\mathcal L=(L,\preceq)$ is a complete lattice
and $g \colon \mathcal L \to \mathcal L$ is a monotone mapping,
then $g$ has at least one fixed point (i.e.,
$x\in L$ with $g(x)=x$), and moreover, there
exists a smallest fixed point $x_0$ and a largest fixed point $x_1$, i.e.,
such that $x_0\preceq x\preceq x_1$
for every fixed point $x$.} 
To prove Theorem~\ref{t:doublezone},
we let $L$ be the system of all ordered $n$-tuples $\DD$
such that $P _i \subseteq D _i$ for every $i$. 
We introduce the ordering $\preceq$ as above 
(one has to check that this gives
a complete lattice, which is straightforward). 
Let $g:=\Dom^2$; that is, $
g (\DD) := \Dom (\Dom \DD)
$.  Then we let $\RR$ be the smallest fixed point of $g$
as in the Knaster--Tarski theorem, and $\SS:=\Dom\RR$. 
Clearly $\Dom\SS=\Dom^2 \RR=g(\RR)=\RR$. 
Moreover, if $\RR',\SS'$ satisfy
$\RR'=\Dom \SS'$ and $\SS'=\Dom\RR'$, then
$\RR'$ and $\SS'$ are both fixed points of $\Dom^2$,
and thus $\RR\preceq\RR',\SS'\preceq \SS$ as claimed.

\section{Proof of Lemma~\ref{l:(1)norm}}
\label{appendix section: inflated norm lemma}

The construction has two positive parameters, $\alpha$ and
$\delta$, where $\alpha$ is small and $\delta$ is still much smaller.

We let $\| \mathord\cdot \|'$ be the Euclidean norm scaled by $\alpha$ in the horizontal
direction; that is, $\|(x,y)\|'=\sqrt{\alpha^2 x^2+y^2}$. 
Let $\| \mathord\cdot \|''$ be the $\ell_1$ norm scaled by a suitable factor
$\beta$ (close to $1$) in the vertical direction:
$\|(x,y)\|''=|x|+\beta|y|$. The norm $\| \mathord\cdot \|_{(1)}$ is obtained
as  $a'\| \mathord\cdot \|'+a''\| \mathord\cdot \|''$, where $a',a''>0$ are suitable
 coefficients.
This obviously yields a norm,
which is rotund since $\| \mathord\cdot \|'$ is rotund. 

We want that the contribution of $\| \mathord\cdot \|'$ is small compared
to that of $\| \mathord\cdot \|''$, and that 
the corners of the unit ball of $\| \mathord\cdot \|_{(1)}$ 
coincide with those of the $\ell_1$ unit ball.
This finally leads to the formula
$$
\|(x,y)\|_{(1)}:= \delta\sqrt{\alpha^2 x^2+y^2}+(1-\alpha\delta)|x|+(1-\delta)|y|.
$$
Fig.~\ref{f:bulgefat} is actually obtained from this formula
with $\delta=0.7$ and $\alpha=0.5$. It is easy to check that, as
the picture suggests, $\| \mathord\cdot \|_{(1)}\le \| \mathord\cdot \|_1$
(and thus the $\ell_1$ unit ball is contained in the
$\| \mathord\cdot \|_{(1)}$ unit ball), and 
for $\delta$ is sufficiently small in terms of $\alpha$ and $\eps$,
the unit ball of $\| \mathord\cdot \|_{(1)}$ is contained in the octagon as in the lemma.

It remains to investigate the bisector of $p$ and $q$ for $x\ge1$ and $y\ge 1$.
For convenience, we translate $p$ and $q$ by $(-1,-1)$ and scale by
$\frac 12$. Then the bisector 
 is given by the equation $\|(x+1,y)\|_{(1)}=\|(x,y+1)\|_{(1)}$,
with the region of interest being the positive quadrant $x,y\ge 0$.
For $x,y\ge 0$, the absolute values can be removed, $\delta$ disappears
from the equation, and we obtain
$\sqrt{\alpha^2 (x+1)^2+y^2}+1-\alpha=\sqrt{\alpha^2 x^2+(y+1)^2}
$. This can be solved for $y$ explicitly, with the only positive root
$$
y=\frac {1-\alpha}{2-\alpha}\left( \sqrt{1+2\alpha x+2\alpha x^2}-1
+\frac{\alpha}{1-\alpha} x\right).
$$
This is the equation of the bisector curve in the positive quadrant.
It is a simple exercise in calculus (distinguishing the cases 
$\alpha x\le 1$ and $\alpha x>1$, say) to show that
$y\le C\sqrt\alpha\, x$ for all $x>0$ and all sufficiently small $\alpha$ (here $C$ is a suitable constant).


\begin{thebibliography}{11}

\bibitem{AsanoKirkpatrick}
T.~Asano and D.~Kirkpatrick. Distance trisector curves in regular 
convex distance metrics. In 
\emph{Proc. 3rd International Symposium on Voronoi Diagrams in Science 
and Engineering}, 
IEEE Computer Society, 
pages 8--17, 2006.

\bibitem{asano07:_zone_diagr}
T.~Asano, J.~Matou{\v{s}}ek, and T.~Tokuyama.
\newblock Zone diagrams: Existence, uniqueness, and algorithmic challenge.
\newblock {\em SIAM Journal on Computing}, 37(4):1182--1198, 2007.

\bibitem{asano07:_distan_trisec_curve}
T.~Asano, J.~Matou{\v{s}}ek, and T.~Tokuyama.
\newblock The distance trisector curve.
\newblock {\em Advances in Mathematics}, 212(1):338--360, 2007.

\bibitem{aurenhammer91:_voron_diagr}
F.~Aurenhammer.
\newblock {V}oronoi diagrams---a survey of a fundamental geometric data
  structure.
\newblock {\em ACM Computing Surveys}, 23(3):345--405, 1991.

\bibitem{BenyaminiLindenstrauss}
Y.~Benyamini and J.~Lindenstrauss.
\newblock {\em Nonlinear Functional Analysis, Vol.~I, Colloquium Publications
  48}.
\newblock American Mathematical Society (AMS), Providence, RI, 1999.


\bibitem{chun07:_distan_trisec_of_segmen_and}
J.~Chun, Y.~Okada, and T.~Tokuyama.
\newblock Distance trisector of segments and zone diagram of segments in a
  plane.
\newblock In 
{\em Proc. 4th International Symposium on Voronoi
Diagrams in Science and Engineering}, 
IEEE Computer Society, 
pages 66--73, 2007. 

\bibitem{okabe00:_spatial_tessel}
A.~Okabe, B.~Boots, K.~Sugihara, and S.~N. Chiu.
\newblock {\em Spatial Tessellations: Concepts and Applications of {V}oronoi Diagrams}.
\newblock Probability and Statistics. Wiley, second edition, 2000.

\bibitem{reem07:_zone_and_doubl_zone_diagr}
D.~Reem and S.~Reich.
\newblock Zone and double zone diagrams in abstract spaces.
\emph{Colloquium Mathematicum} 115(1):129--145, 2009.

\bibitem{tarski55:_lattic_theor_fixpoin_theor_and_its_applic}
A.~Tarski.
\newblock A lattice-theoretical fixpoint theorem and its applications.
\newblock {\em Pacific Journal of Mathematics}, 5:285--309, 1955.
\end{thebibliography}
\end{document}